\documentclass{llncs}
\usepackage[dvips]{graphicx}
\usepackage{amssymb}
\usepackage{amsmath}
\usepackage{xcolor}
\usepackage{verbatim}

\def\Q{\mathcal{Q}}
\def\P{\mathcal{P}}
\def\R{\mathcal{R}}
\def\T{\mathcal{T}}
\def\U{\mathcal{U}}
\def\Y{\mathcal{Y}}

\def\MinDel{\textup{MinDel}}
\def\mydef{\textup{def}}
\def\mysize{\textup{size}}
\def\TableComp{\textup{SolTable}}

\begin{document}

\title{Obtaining a Proportional Allocation by Deleting Items}
\titlerunning{Proportionality by Deleting Items}

\author{
  Britta Dorn\inst{1} 
\and 
  Ronald de Haan\inst{2}
\and
  Ildik\'o Schlotter\inst{3}
}
\institute{
  {University of T\"ubingen, Germany,
  \texttt{britta.dorn@uni-tuebingen.de}}
\and  
  {University of Amsterdam, the Netherlands,
  \texttt{R.deHaan@uva.nl}}
\and
  {Budapest University of Technology and Economics, Hungary,
  \texttt{ildi@cs.bme.hu}}
}

\maketitle

\begin{abstract}
We consider the following control problem on fair allocation of indivisible goods. Given a set $I$ of items and a set of agents, each having strict linear preference over the items, we ask for a minimum subset of the items whose deletion guarantees the existence of a proportional allocation in the remaining instance; we call this problem \textsc{Proportionality by Item Deletion (PID)}. Our main result is a polynomial-time algorithm that solves PID for three agents. By contrast, we prove that PID is computationally intractable when the number of agents is unbounded, even if the number $k$  of item deletions allowed is small, since the problem turns out to be $\mathsf{W}[3]$-hard with respect to the parameter $k$. Additionally, we provide some tight lower and upper bounds on the complexity of PID when regarded as a function of $|I|$ and $k$. 
\end{abstract}

\section{Introduction}
\label{sec:intro}

We consider a situation where a set $I$ of indivisible items needs to be allocated to a set $N$ of agents in a way that is perceived as \emph{fair}. Unfortunately, it may happen that a fair allocation does not exist in a setting. In such situations, we might be interested in the question how our instance can be modified in order to achieve a fair outcome. Naturally, we seek for a modification that is as small as possible. This can be thought of as a \emph{control action} carried out by a central agency whose task is to find a fair allocation. The computational study of such control problems was first proposed by Bartholdi, III et al.~\cite{bartholdi-control} for voting systems; our paper follows the work of Aziz et al.~\cite{ASW-ijcai} who have recently initiated the systematic study of control problems in the area of fair division.

The idea of fairness can be formalized in various different ways such as proportionality, envy-freeness, or max-min fair share. Here we focus on \emph{proportionality}, a notion originally defined in a model where agents use utility functions to represent their preferences over items. In that context, an allocation is called proportional if each agent obtains a set of items whose utility is at least $1/|N|$ of their total utility of all items. One way to adapt this notion to a model with linear preferences (not using explicit utilities) is to look for an allocation that is proportional with respect to \emph{any} choice of utility functions for the agents that is compatible with the given linear preferences (see Aziz et al.~\cite{AGMW-2015} for a survey of other possible notions of proportionality and fairness under linear preferences). 
Aziz et al.~\cite{AGMW-2015} referred to this property as ``necessary proportionality''; for simplicity, we use the shorter term ``proportionality.''

We have two reasons for considering linear preferences. First, an important  advantage of this setting is the easier elicitation of agents' preferences, which enables for more practical applications. Second, this simpler model is more tractable in a computational sense: under linear preferences, the existence of a proportional allocation can be decided in polynomial time~\cite{AGMW-2015}, whereas the same question for cardinal utilities is \textsf{NP}-hard already for two agents~\cite{lipton-proportionality}. Clearly, if already the existence of a proportional allocation is computationally hard to decide, then we have no hope to solve the corresponding control problem efficiently.

Control actions can take various forms. Aziz et al.~\cite{ASW-ijcai} mention several possibilities: control by adding/deleting/replacing agents or items in the given instance, or by partitioning the set of agents or items. In this paper we concentrate only on control by \emph{item deletion}, where the task is to find a subset of the items, as small as possible, whose removal from the instance guarantees the existence of a proportional allocation. 
In other words, we ask for the maximum number of items that can be allocated to the agents in a proportional way. 

\subsection{Related Work}

We follow the research direction proposed by Aziz et al.~\cite{ASW-ijcai} who initiated the study of control problems in the area of fair division. 
As an example, Aziz et al.~\cite{ASW-ijcai} consider the complexity of obtaining envy-freeness by adding or deleting items or agents, assuming linear preferences. They show that adding/deleting a minimum number of items to ensure envy-freeness can be done in polynomial time for two agents, while for three agents it is \textsf{NP}-hard even to decide if an envy-free allocation exists. As a consequence, they obtain \textsf{NP}-hardness also for the control problems where we want to ensure envy-freeness by  adding/deleting items in case there are more than two agents, or by adding/deleting agents. 

The problem of deleting a minimum number of items to obtain envy-freeness was first studied by Brams et al.~\cite{brams-AL} who gave a polynomial-time algorithm for the case of two agents.\footnote{
For a complete proof of the correctness of their algorithm, see also~\cite{ASW-ijcai}.
} In the context of cake cutting, Segal-Halevi et al.~\cite{segal-halevi-cakecutting} proposed the idea of distributing only a portion of the entire cake in order to obtain an envy-free allocation efficiently. 
For the  Hospitals/Residents with Couples problem, Nguyen and Vohra~\cite{nguyen-vohra-hrc} considered another type of control action: they obtained stability by slightly perturbing the capacities of hospitals.
%The  Hospitals/Residents with Couples matching problem can also be thought of as a fair division problem: Teaching positions in hospitals have to be allocated to medical students in a fair (e.g.\ stable) way, taking into account the joint preferences of couples. For this problem, Nguyen and Vohra~\cite{nguyen-vohra-hrc} considered another type of control action: they obtained stability by slightly perturbing the capacities of hospitals.

\subsection{Our Contribution}
We first consider the case where the number of agents is unbounded (see Section~\ref{sec:unbounded-agents}). We show that the problem of deciding whether there exist at most $k$ items whose deletion allows for a proportional allocation is \textsf{NP}-complete, and also $\mathsf{W}[3]$-hard with parameter~$k$ (see Theorem~\ref{thm-w3-hardness}). 
This latter result shows that even if we allow only a few items to be deleted, we cannot expect an efficient algorithm, since the problem is not fixed-parameter tractable with respect to the parameter $k$ (unless $\mathsf{FPT} = \mathsf{W}[3]$).

Additionally, we provide tight upper and lower bounds on the complexity of the problem. 
In Theorem~\ref{thm-linear-w2-hardness} we prove that the trivial $|I|^{O(k)}$ time algorithm---that, in a brute force manner, checks for each subset of $I$ of size at most $k$ whether it is a solution---is essentially optimal (under the assumption $\mathsf{FPT} \neq \mathsf{W}[1]$). We provide another simple algorithm in Theorem~\ref{tight-exp-algo} that has optimal running time, assuming the Exponential Time Hypothesis. 

In Section~\ref{sec:threeagents}, we turn our attention to the case with only three agents. In Theorem~\ref{thm-3agents-poly} we propose a polynomial-time algorithm for this case, which can be viewed as our main result. The presented algorithm is based on dynamic programming, but relies heavily on a non-trivial insight into the structure of solutions. 

\section{Preliminaries}
We assume the reader to be familiar with basic complexity theory, in particular with parameterized complexity~\cite{DowneyFellows13}.

\medskip
\noindent
{\bf Preferences.}
Let $N$ be a set of agents and $I$ a set of indivisible items that we wish to allocate to the agents in some way. We assume that each agent $a \in N$ has strict 
preferences over the items, expressed by a preference list $L^a$ that is a linear ordering of $I$, and set $L=\{L^x\mid x \in N\}$. We call the triple $(N,I,L)$ a \emph{(preference) profile}. 
We denote by $L^a[i:j]$ the subsequence of $L^a$ containing the items ranked by agent $a$ between the positions $i$ and $j$, 
inclusively, for any $1 \leq i \leq j \leq |I|$.  Also, for a subset~$X \subseteq I$ of items 
we denote by~$L^a_X$ the restriction of~$L^a$ to the items in~$X$. 

\medskip
\noindent
{\bf Proportionality.}
Interestingly, the concept of proportionality (as described in Section~\ref{sec:intro}) has an equivalent definition that is more direct and practical: we say that an allocation $\pi: I \rightarrow N$ mapping items to agents is \emph{proportional} if for any integer $i \in \{1, \dots, |I|\}$ and any agent $a \in N$, the number of items from~$L^a[1:i]$ allocated to~$a$ by~$\pi$ is at least $i/|N|$. 
Note that, in particular, this means that in a proportional allocation, each agent needs to get his or her first choice. Another important observation is that a proportional allocation can only exist if the number of items is a multiple of $|N|$, since each agent needs to obtain at least $|I|/|N|$ items. 
%There are $|I|$ items to be allocated in total, hence $|N| \cdot \lceil|I|/|N|\rceil = |I|$ must hold. This is only possible if $\lceil|I|/|N|\rceil = |I|/|N|$, which means that the number of items has to be a multiple of $|N|$. 

\medskip
\noindent
{\bf Control by deleting items.}
Given a profile $\mathcal{P}=(N,I,L)$ and a subset $U$ of items, we can define the preference profile  $\mathcal{P}-U$ obtained by removing all items in~$U$ from $I$ and from all preference lists in $L$.
Let us define the \textsc{Proportionality by Item Deletion (PID)} problem as follows.
Its input is a pair $(\mathcal{P},k)$ where $\mathcal{P}=(N,I,L)$ is a preference profile and $k$ is an integer. 
We call a set $U \subseteq I$ of items a \emph{solution} for $\P$
if its removal from $I$ allows for proportionality, that is, 
if there exists a proportional allocation $\pi:I \setminus U \rightarrow N$ 
for $\mathcal{P}-U$. 
The task in PID is to decide if there exists a solution of size at most $k$.

\section{Unbounded Number of Agents}
\label{sec:unbounded-agents}

Since the existence of a proportional allocation can be decided in polynomial time 
by standard techniques in matching theory \cite{AGMW-2015}, 
the \textsc{Proportional Item Deletion} problem is solvable in $|I|^{O(k)}$ time 
by the brute force algorithm that checks for each subset of $I$ 
of size at most $k$ whether it is a solution. 
In terms of parameterized complexity, this means that PID is in $\mathsf{XP}$
when parameterized by the solution size. 

Clearly, such a brute force approach may only be feasible if the number~$k$ of items we are allowed to delete is very small. Searching for a more efficient algorithm, one might ask whether the problem becomes fixed-parameter tractable with $k$ as the parameter, 
i.e., whether there exists an algorithm for PID that, for an instance $(\mathcal{P},k)$ runs in time $f(k) |\mathcal{P}|^{O(1)}$ for some computable function $f$. Such an algorithm could be much faster in practice compared to the brute force approach described above. 

Unfortunately, the next theorem shows that finding such a fixed-parameter tractable algorithm seems unlikely, as PID is $\mathsf{W}[2]$-hard with parameter $k$. Hence, deciding whether the deletion of $k$ items can result in a profile admitting a proportional allocation is computationally intractable even for small values of $k$. 

\begin{theorem}
\label{thm-w2hard}
\textsc{Proportionality by Item Deletion} is \textsf{NP}-complete and $\mathsf{W}[2]$-hard when parameterized by the size $k$ of the desired solution. 
\end{theorem}

\begin{proof}
We are going to present an FPT-reduction from the $\mathsf{W}[2]$-hard problem \textsc{$k$-Dominating Set}, 
where we are given a graph $G=(V,E)$ and an integer $k$ and the task is to decide if 
$G$ contains a dominating set of size at most $k$; a vertex set $D \subseteq V$ is \emph{dominating} in $G$ if each vertex in $V$ is either in $D$ or has a neighbor in $D$. We denote by $N(v)$ the set of neighbors of some vertex $v \in V$, and we let $N[v]=N(v) \cup \{v\}$. Thus, a vertex set $D$ is dominating if $N[v] \cap D \neq \emptyset$ holds for each $v \in V$. 

Let us construct an instance $I_{\textup{PID}}=(\mathcal{P},k)$ of PID with $\mathcal{P}=(N,I,L)$ as follows. 
We let $N$ contain $3n+2m+1$ agents where $n=|V|$ and $m=|E|$: 
we create~$n+1$ so-called \emph{selection agents} $s_1, \dots, s_{n+1}$,
and for each $v \in V$ we create a set $A_v=\{a_v^j \mid 1 \leq j \leq |N[v]|+1\}$
of \emph{vertex agents}.
Next we let $I$ contain $2|N|+k$ items: we create 
distinct first-choice items $f(a)$ for each agent $a \in N$, 
a \emph{vertex item} $i_v$ for each $v \in V$,
a dummy item $d_v^j$ for each vertex agent $a_v^j \in N$, 
and~$k+1$ additional dummy items $c_1, \dots, c_{k+1}$.

Let $F$ denote the set of all first-choice items, i.e., $F=\{f(a) \mid a \in N\}$.
For any set $U \subseteq V$ of vertices in $G$, let $I_U=\{i_v \mid v \in U\}$; 
in particular, $I_V$ denotes the set of all vertex items.

Before defining the preferences of agents, we need some additional notation. 
We fix an arbitrary ordering $\prec$ over the items, and 
for any set $X$ of items we let $[X]$ denote the ordering of $X$ according to $\prec$. 
Also, for any $a \in N$, we define the set $F^a_i$ as the first $i$ elements of $F \setminus \{f(a)\}$, 
for any $i \in \{1, \dots, |N|-1\}$. 
% $\binom{X}{i}$ denote an arbitrary size-$i$ subset of $X$ for any $1 \leq i \leq |X|$.
We end preference lists below with the symbol `$\dots$' meaning 
all remaining items not listed explicitly, ordered according to $\prec$.

Now we are ready to define the preference list $L^a$ for each agent $a$. 
%First, we let the preference list of a selection agent $a=s_i$ be as follows.
\begin{itemize}
\item If $a$ is a selection agent $a=s_i$ with $1 \leq i \leq n-k$, then let
$$L^{a}: 
		f(a), \underbrace{[F^a_{|N|-n}], [I_V]}_{\textrm{\scriptsize $|N|$ items}}, 
		\underbrace{[F^a_{|N|-n+k} \setminus  F^a_{|N|-n}]}_{\textrm{\scriptsize $k$ items}}, \dots $$
\item If $a$ is a selection agent $a=s_i$ with $n-k < i \leq n+1$, then let 
$$L^{a}: 
		f(a), \underbrace{[F^a_{|N|-n}], [I_V]}_{\textrm{\scriptsize $|N|$ items}}, 
		\underbrace{[F^a_{|N|-n+k-1} \setminus  F^a_{|N|-n}]}_{\textrm{\scriptsize $k-1$ items}}, c_{i-(n-k)}, \dots $$
\item If $a$ is a vertex agent $a=a_v^j$ with $1 \leq j \leq |N[v]|+1$, then let
$$L^{a}: 
       f(a), \underbrace{[F^a_{|N|-|N[v]|}], [I_{N[v]}]}_{\textrm{\scriptsize $|N|$ items}}, 
       d_v^j, \dots $$
\end{itemize}
This finishes the definition of our PID instance $I_{\textup{PID}}$.

Suppose that there exists a solution $S$ of size at most $k$ to $I_{\textup{PID}}$ 
and a proportional allocation $\pi$ mapping the items of $I \setminus S$ to the agents in $N$. 
Observe that by $|I|=2|N|+k$, we know that $S$ must contain exactly $k$ items. 

First, we show that $S$ cannot contain any item from $F$. 
For contradiction, assume that $f(a) \in S$ for some agent $a$. Since the preference list of $a$
starts with more than $k$ items from $F$ (by $N-n>k$), the first item in $L^a_{I \setminus S}$ 
must be an item $f(b)$ for some $b \in N$, $b \neq a$. The first item in $L^b_{I \setminus S}$ 
is exactly $f(b)$, and thus any proportional allocation should allocate 
$f(b)$ to both~$a$ and~$b$, a contradiction.

Next, we prove that $S \subseteq I_V$. For contradiction, 
assume that $S$ contains less than $k$ items from $I_V$. 
Then, after the removal of $S$, the top $|N|+1$ items in the preference list 
$L^{s_i}_{I \setminus S}$ of any selection agent $s_i$ are all contained in $I_V \cup F$. 
Hence, $\pi$ must allocate at least two items from $I_V \cup F$ to $s_i$, by the definition of proportionality. 
Recall that for any agent $a$, $\pi$ allocates $f(a)$ to $a$, meaning that~$\pi$ 
would need to distribute the $n$ items in $I_V$ among the $n+1$ selection agents, a contradiction. 
Hence, we have $S \subseteq I_V$. 

We claim that the $k$ vertices $D=\{v \mid i_v \in S\}$ form a dominating set in $S$. 
Let us fix a vertex $v \in V$. For sake of contradiction, 
assume that $N[v] \cap D =\emptyset$, and consider any vertex agent $a$ in $A_v$. 
Then the top $|N|+1$ items in $L^a_{I \setminus S}$
are the same as the top $|N|+1$ items in $L^a=L^a_I$ (using that $S \cap F=\emptyset$), 
and these items form a subset of $I_{N[v]} \cup F$ for every $a \in A_v$. 
But then arguing as above, we get that $\pi$ would need to allocate an item of $I_{N[v]}$ to
each of the $|N[v]|+1$ vertex agents in $A_v$; again a contradiction. 
Hence, we get that $N[v] \cap D \neq \emptyset$ for each $v \in V$, showing that $D$ is 
indeed a dominating set of size $k$. 

For the other direction, let $D$ be a dominating set of size $k$ in $G$, 
and let $S$ denote the set of $k$ vertex items $\{i_v \mid v \in D\}$. 
To prove that $S$ is a solution for~$I_{\textup{PID}}$, we define a proportional allocation $\pi$ in 
the instance obtained by removing~$S$. 
First, for each selection agent $s_i$ with $1 \leq i \leq n-k$, we let $\pi$ 
allocate $f(s_i)$ and the $i$th item from $I_{V \setminus D}$ to $s_i$ .
Second, for each selection agent $s_{n-k+i}$ with $1 \leq i \leq k+1$, we let $\pi$ 
allocate $f(s_{n-k+i})$ and the dummy item $c_i$ to $s_{n-k+i}$.
Third, $\pi$ allocates the items $f(a_v^j)$ and $d_v^j$ to each vertex agent $a_v^j \in N$.

It is straightforward to check that $\pi$ is indeed proportional. 

For proving $\mathsf{NP}$-completeness, observe that the presented FPT-reduction is a polynomial reduction as well, so the $\mathsf{NP}$-hardness of \textsc{Dominating Set} implies that PID is $\mathsf{NP}$-hard as well; since for any subset of the items we can verify in polynomial time whether it yields a solution, containment in $\mathsf{NP}$ follows.
\qed
\end{proof}

In fact, we can strengthen the $\mathsf{W}[2]$-hardness result of Theorem~\ref{thm-w2hard} and show that PID is even $\mathsf{W}[3]$-hard with respect to parameter $k$.\footnote{We present Theorem~\ref{thm-w2hard} so that we can re-use its proof for Theorems~\ref{thm-linear-w2-hardness} and \ref{tight-exp-algo}.
 }

\begin{theorem}
\label{thm-w3-hardness}
\textsc{Proportionality by Item Deletion} is $\mathsf{W}[3]$-hard when parameterized by the size $k$ of the desired solution. \end{theorem}
\begin{proof}
We are going to present an FPT-reduction from the $\mathsf{W}[3]$-hard \textsc{wcs}$^-$[3] problem, which is the weighted
satisfiability problem for formulas of the
form $\varphi = \bigwedge\nolimits_{i=1}^{m_1} \bigvee\nolimits_{j=1}^{m_{2,i}}
\bigwedge\nolimits_{\ell=1}^{m_{3,i,j}} l_{i,j,\ell}$,
where each $l_{i,j,\ell}$ is a negative literal
\cite{DowneyFellows95,FlumGrohe05,FlumGrohe06}.

Let $(\varphi,k)$ be an instance of the weighted satisfiability problem,
where $\varphi$ is a formula of the form described above.
Let $X = \{ x_1,\dotsc,x_n \}$ be the set of all variables occurring in $\varphi$---%
that is, $n$ denotes the number of variables in $\varphi$.
We will construct an instance $I_{\textup{PID}}=(\mathcal{P},k)$ of PID
with $\mathcal{P}=(N,I,L)$ as follows. 
We let $N$ contain $n+1+\sum\nolimits_{i = 1}^{m_1} m_{2,i}$ agents: 
we create $n+1$ so-called \emph{selection agents} $s_1,\dotsc,s_{n+1}$,
and for each $1 \leq i \leq m_1$
we create a set $A_i =\{a_i^j \mid 1 \leq j \leq m_{2,i}\}$
of \emph{verification agents}.
Next we let $I$ contain $2|N|+k$ items: we create 
distinct \emph{first-choice items} $f(a)$ for each agent $a \in N$, 
a \emph{variable item} $w_u$ for each $1 \leq u \leq n$,
$m_{2,i}$ \emph{verification items} $y_{i,1},\dotsc,y_{i,m_{2,i}}$ for each $1 \leq i \leq m_1$,
and $k+1$ \emph{dummy items} $c_1, \dots, c_{k+1}$.

Let $F$ denote the set of all first-choice items, i.e., $F=\{f(a) \mid a \in N\}$.
For any subset $X' \subseteq X$ of variables, let $W_{X'}=\{w_u \mid x_u \in X' \}$; 
in particular, $W_X$ denotes the set of all variable items. 

Before defining the preferences of agents, we need some additional notation. 
We fix an arbitrary ordering $\prec$ over the items, and 
for any set $Z$ of items we let $[Z]$ denote the ordering of $Z$ according to $\prec$. 
Also, for any $a \in N$, we define the set $F^a_i$ as the first $i$ elements of $F \setminus \{f(a)\}$, 
for any $i \in \{1, \dots, |N|-1\}$. 
Moreover, for any $1 \leq i \leq m_1$ we define the sets
$Y_i = \{ y_{i,1},\dotsc,y_{i,m_{2,i}} \}$
and $Y'_i = \{ y_{i,1},\dotsc,y_{i,m_{2,i}-1} \}$.
We end preference lists below with the symbol `$\dots$' meaning 
all remaining items not listed explicitly, ordered according to $\prec$.

Now we are ready to define the preference list $L^a$ for each agent $a$. 
\begin{itemize}
%%%
\item If $a$ is a selection agent $a=s_i$ with $1 \leq i \leq n-k$, then let
$$L^{a}:
f(a), \underbrace{[F^a_{|N|-n}], [W_X]}_{\textrm{\scriptsize $|N|$ items}}, 
\underbrace{[F^a_{|N|-n+k} \setminus  F^a_{|N|-n}]}_{\textrm{\scriptsize $k$ items}},
\dots $$
%%%
\item If $a$ is a selection agent $a=s_i$ with $n-k < i \leq n+1$, then let 
$$L^{a}: 
f(a), \underbrace{[F^a_{|N|-n}], [W_X]}_{\textrm{\scriptsize $|N|$ items}}, 
\underbrace{[F^a_{|N|-n+k-1} \setminus  F^a_{|N|-n}]}_{\textrm{\scriptsize $k-1$ items}}, c_{i-(n-k)},
\dots $$
%%%
\item If $a$ is a verification agent $a=a_i^j$ for $1 \leq i \leq m_1$
and $1 \leq j \leq m_{2,i}$, then let
$$L^{a}: 
f(a), \underbrace{[F^a_{|N|-|C_{i,j}|-|Y'_i|+k}], [Y'_i], [C_{i,j}]}_{\textrm{\scriptsize $|N|+k$ items}},
y_{i,m_{2,i}},
\dots $$
where $C_{i,j} = X \setminus \{ x \in X \mid 1 \leq \ell \leq m_{3,i,j}, l_{i,j,\ell} = \neg x \}$
is the set of variables that do not occur in any literal $l_{i,j,\ell}$,
for $1 \leq \ell \leq m_{3,i,j}$.
%%%
\end{itemize}
This finishes the definition of our PID instance $I_{\textup{PID}}$.

Suppose that there exists a solution $S$ of size at most $k$ to $I_{\textup{PID}}$ 
and a proportional allocation $\pi$ mapping the items of $I \setminus S$ to the agents in $N$. 
Observe that by $|I|=2|N|+k$, we know that $S$ must contain exactly $k$ items. 

First, we show that $S$ cannot contain any item from $F$. 
To derive a contradiction, assume that $f(a) \in S$ for some agent $a$.
We can safely assume that $|N|-n>k$
and that $|N| - n > m_{2,i}$ for each $1 \leq i \leq m_1$.
As a result, the preference list of $a$
starts with more than $k$ items from $F$.
Therefore, the first item in $L^a_{I \setminus S}$ 
must be an item~$f(b)$ for some $b \in N$, $b \neq a$.
Clearly, the first item in $L^b_{I \setminus S}$ 
is exactly $f(b)$, which means that any proportional allocation should allocate 
$f(b)$ to both $a$ and $b$, which is a contradiction.

Next, we prove that $S \subseteq W_X$. To derive a contradiction, 
assume that $S$ contains less than $k$ items from $W_X$. 
Then, after the removal of $S$, the top $|N|+1$ items in the preference list 
$L^{s_i}_{I \setminus S}$ of any selection agent $s_i$ are all contained in $W_X \cup F$. 
Hence, $\pi$ must allocate at least two items from $W_X \cup F$ to each $s_i$,
by the definition of proportionality. 
Recall that for any agent $a$, $\pi$ allocates $f(a)$ to $a$, meaning that $\pi$ 
would need to distribute the $n$ items in $W_X$ among
the $n+1$ selection agents, which is a contradiction. 
Hence, we have $S \subseteq W_X$. 
We also get that $\pi$ must allocate all items 
in $W_X \setminus S \cup \{c_1, \dots, c_{k+1}\}$ to the selection agents.

Consider the truth assignment $\alpha : X \rightarrow \{ 0,1 \}$
defined by letting $\alpha(x_u) = 1$ if and only if $w_u \in S$,
for each $x_u \in X$.
Since $|S| = k$, the truth assignment $\alpha$ has weight $k$.
We show that $\alpha$ satisfies $\varphi$.
To do so, we need to show that for each $1 \leq i \leq m_1$
it holds that $\alpha$ satisfies
$\varphi_i = \bigvee\nolimits_{j=1}^{m_{2,i}}
\bigwedge\nolimits_{\ell=1}^{m_{3,i,j}} l_{i,j,\ell}$.
Take an arbitrary $1 \leq i \leq m_1$.
To derive a contradiction, assume that for each $1 \leq j \leq m_{2,i}$
it holds that there is some $1 \leq \ell \leq m_{3,i,j}$
such that $l_{i,j,\ell}$ is made false by $\alpha$.
Then for each such $1 \leq j \leq m_{2,i}$ it holds
that $|C_{i,j} \cap S| < k$.
Then for each verification agent $a_i^j$, for $1 \leq j \leq m_{2,i}$
it holds that the top $|N|+1$ items in $L^{a}_{I \setminus S}$ (for $a = a_i^j$)
form a subset of $Y'_i \cup W_X \cup F$.
Then arguing as above, we get that $\pi$ would need to allocate an item of $Y'_i$ to
each of the $|Y_i| = |Y'_i|+1$ agents $a_i^j$, which is a contradiction.
Since $i$ was arbitrary, we can conclude that $\alpha$
satisfies $\varphi$.

For the other direction, let $\alpha : X \rightarrow \{ 0,1 \}$ be a truth assignment
of weight $k$ that satisfies $\varphi$,
and let $S$ denote the set of $k$ variable items $\{ w_u \mid x_u \in X, \alpha(x_u) = 1 \}$.
To prove that $S$ is a solution for $I_{\textup{PID}}$, we define a proportional allocation $\pi$ in 
the instance obtained by removing $S$.
First, for each selection agent $s_i$ with $1 \leq i \leq n-k$, we let $\pi$ 
allocate $f(s_i)$ and the $i$th item from $W_{X} \setminus S$ to $s_i$.
Second, for each selection agent $s_{n-k+i}$ with $1 \leq i \leq k+1$, we let $\pi$ 
allocate $f(s_{n-k+i})$ and the dummy item $c_i$ to $s_{n-k+i}$.
Then, for each $1 \leq i \leq m_{i}$, let $1 \leq j_i \leq m_{2,i}$ be some
number such that $\alpha$ makes $\bigwedge\nolimits_{\ell = 1}^{m_{3,i,j_i}} l_{i,j_i,\ell}$
true---we know that such a $j_i$ exists for each $i$ because $\alpha$
satisfies $\varphi$.
For each verification agent $a_i^j$ we let $\pi$ allocate $f(a_i^j)$ and one item
from $Y_i$ to $a_i^j$ as follows.
If $j = j_i$, we let $\pi$ allocate $y_{i,m_{2,i}}$ to $a_i^j$;
if $j < j_i$, we let $\pi$ allocate $y_{i,j}$ to $a_i^j$; and
if $j > j_i$, we let $\pi$ allocate $y_{i,j-1}$ to $a_i^j$.
It is straightforward to check that $\pi$ is indeed proportional. 
\qed
\end{proof}

Theorem~\ref{thm-w3-hardness} implies that we cannot expect an FPT-algorithm for PID with respect to the parameter $k$, the number of item deletions allowed, unless $\textsf{FPT} \neq \textsf{W[3]}$. Next we show that the brute force algorithm that runs in $|I|^{O(k)}$ time is optimal, assuming the slightly stronger assumption $\textsf{FPT} \neq \textsf{W[1]}$.

\begin{theorem}
\label{thm-linear-w2-hardness}
There is no algorithm for PID that on an instance $(\mathcal{P},k)$ with item set $I$ runs in $f(k) |I|^{o(k)} |\mathcal{P}|^{O(1)}$ time for some function $f$, unless $\mathsf{FPT} \neq \mathsf{W}[1]$.\footnote{Here, we use an effective variant of ``little o'' (see, e.g. \cite[Definition~3.22]{FlumGrohe06}).} 
\end{theorem}

\begin{proof}
Chen et al.~\cite{chen-w1-vs-eth-stronger} introduced the class of 
$\mathsf{W}_l[2]$-hard problems based on the notion of \emph{linear FPT-reductions}. They proved that \textsc{Dominating Set} is $\mathsf{W}_l[2]$-hard, and that this implies a strong lower bound on its complexity: unless $\mathsf{FPT} \neq \mathsf{W}[1]$, \textsc{Dominating Set} cannot be solved in $f(k) |V|^{o(k)} (|V|+|E|)^{O(1)}$ time for any function $f$.

Observe that in the FPT-reduction presented in the proof of Theorem~\ref{thm-w2hard} the new parameter has linear dependence on the original parameter (in fact they coincide). Therefore, this reduction is a linear FPT-reduction, and consequentially, PID is $\mathsf{W}_l[2]$-hard. Hence, as proved by Chen et al.~\cite{chen-w1-vs-eth-stronger}, PID on an instance~$(\mathcal{P},k)$ with item set~$I$ cannot be solved in time $f(k)|I|^{o(k)} |\mathcal{P}|^{O(1)}$ time for any function $f$, unless  $\textsf{FPT} \neq \textsf{W[1]}$.
\qed
\end{proof}

If we want to optimize the running time not with respect to the number~$k$ of allowed deletions but rather in terms of the total number of items, then we can also give the following tight complexity result, under the Exponential Time Hypothesis (ETH). 
This hypothesis, formulated in the seminal paper by Impagliazzo, Paturi, and Zane~\cite{IPZ-eth} says that \textsc{3-Sat} cannot be solved in $2^{o(n)}$ time, where $n$ is the number of variables in the 3-CNF fomular given as input.

\begin{theorem}
\label{tight-exp-algo}
PID can be solved in $O^\star(2^{|I|})$ time, but unless the ETH fails, it cannot be solved in $2^{o(|I|)}$ time, where $I$ is the set of items in the input. 
\end{theorem}

\begin{proof}
The so-called Sparsification Lemma proved by Impagliazzo et al.~\cite{IPZ-eth} implies that assuming the ETH, \textsc{3-Sat} cannot be solved in $2^{o(m)}$ time, where $m$ is the number of clauses in the 3-CNF formula given as input. Since the standard reduction from \textsc{3-Sat} to \textsc{Dominating Set} transforms a 3-CNF formula with $n$ variables and $m$ clauses into an instance $(G,n)$ of \textsc{Dominating Set} such that the graph $G$ has $O(m)$ vertices and maximum degree 3 (see, e.g., \cite{thulasiraman-handbook}), it follows that \textsc{Dominating Set} on a graph $(V,E)$ cannot be solved in $2^{o(|V|)}$ time even on graphs having maximum degree 3, unless the ETH fails. 

Recall that the reduction presented in the proof of Theorem~\ref{thm-w2hard} computes from each instance $(G,k)$  of \textsc{Dominating Set} with $G=(V,E)$ an instance $(\mathcal{P},k)$ of PID where the number of items is $3|V|+2|E|+1$. Hence, assumming that our input graph $G$ has maximum degree~3, we obtain $|I|=O(|V|)$ for the set~$I$ of items in $\mathcal{P}$. Therefore, an algorithm for PID running in $2^{o(|I|)}$ time would yield an algorithm for \textsc{Dominating Set} running in $2^{o(|V|)}$ time on graphs of maximum degree 3, contradicting the ETH. 
\qed
\end{proof}

\section{Three Agents}
\label{sec:threeagents}

It is known that PID for two agents is solvable in polynomial-time: 
the problem of obtaining an envy-free allocation by item deletion is polynomial-time solvable if there are only two agents \cite{ASW-ijcai,brams-AL}; 
since for two agents an allocation is proportional if and only 
it is envy-free \cite{AGMW-2015}, this proves tractability of PID for $|N|=2$ immediately. In this section, we generalize this result by proving that PID is polynomial-time solvable for three agents.

Let us define the underlying graph $G$ of our profile $\mathcal{P}$ of PID as the following bipartite graph. The vertex set of $G$ consists of the set $I$ of items on the one side, and a set $S$ on the other side, containing all pairs of the form~$(x,i)$ where~$x \in N$ is an agent and $i \in \{1, \dots, \lceil |I|/|N| \rceil \}$. 
Such pairs are called \emph{slots}. 
We can think of the slot $(x,i)$ as the place for the $i$th item that agent $x$ receives in some allocation. We say that an item is \emph{eligible} for a slot $(x,i)$, if it is contained in $L^x[1:|N|(i-1)+1]$.
In the graph $G$, we connect each slot with the items that are eligible for it. Observe that any proportional allocation corresponds to a perfect matching in $G$; 
see Lemma~\ref{lem-propalloc} for a proof.

In what follows, we suppose that our profile $\mathcal{P}$ contains three agents, 
so let~$N=\{a,b,c\}$.

\subsection{Basic Concepts: Prefixes and Minimum Obstructions}
Since our approach to solve PID with three agents is to apply dynamic programming, we need to handle partial instances of PID. Let us define now the basic necessary concepts.

\smallskip
\noindent
{\bf Prefixes.}
For any triple $(i_a,i_b,i_c)$ with $1 \leq i_a,i_b,i_c \leq |I|$ we define a \emph{prefix}
 %sub-profile?
$\mathcal{Q}=\mathcal{P}[i_a,i_b,i_c]$ of $\mathcal{P}$ as the triple $(L^a[1:i_a],L^b[1:i_b],L^c[1:i_c]$), listing only the first $i_a$, $i_b$, $i_c$ items in the preference list of agents $a$, $b$, and $c$, respectively.
We call $(i_a,i_b,i_c)$ the \emph{size} of $\Q$ and denote it by $\textup{size}(\Q)$. 
We also define the \emph{suffix}~$\P-\Q$ as the triple $(L^a[i_a+1:|I|],L^b[i_b+1:|I|],L^c[i_c+1:|I|])$, which can be thought of as the remainder of $\P$ after deleting $\Q$ from it.

We say that a prefix $\P_i=\mathcal{P}[i_a,i_b,i_c]$ is \emph{contained in} another prefix $\P_j=\mathcal{P}[j_a,j_b,j_c]$  if $j_x \leq i_x$ for each $x \in N$; the containment is \emph{strict} if $j_x<i_x$ for some $x \in N$. We say that $\P_i$ and $\P_j$ are \emph{intersecting} 
if none of them contains the other; we call the unique largest prefix 
contained both in $\P_i$ and in $\P_j$, 
i.e., the prefix $\P[\min(i_a,j_a),\min(i_b,j_b),\min(i_c,j_c)]$,
their \emph{intersection}, and denote it by $\P_i \cap \P_j$.

For some prefix $\mathcal{Q}=\mathcal{P}[i_a,i_b,i_c]$, 
let $I(\mathcal{Q})$ denote the set of all items appearing in $\mathcal{Q}$, and 
let $S(\mathcal{Q})$ denote the set of all slots appearing in $\mathcal{Q}$, 
i.e., $S(\Q)=\{ (x,i) \mid 1 \leq i \leq \lceil (i_x+2)/3 \rceil, x \in N\} $. 
We also define the graph $G(\mathcal{Q})$ underlying~$\mathcal{Q}$ 
as the subgraph of $G$ induced by all slots and items appearing in $\mathcal{Q}$, 
that is, $G(\mathcal{Q})=G[S(\mathcal{Q}) \cup I(\mathcal{Q})]$. 
We say that a slot is \emph{complete} in $\Q$, if it is connected to the same items in $G(\Q)$ as in $G$; clearly the only slots which may be incomplete are 
the last slots in $\Q$, that is, the slots $(x,\lceil (i_x+2)/3 \rceil )$, $x \in N$.

\smallskip
\noindent
{\bf Solvability.}
We say that a prefix $\mathcal{Q}$ is \emph{solvable}, if the underlying graph $G(\Q)$ has a matching that covers all its complete slots. 
Hence, a prefix is solvable exactly 
if there exists an allocation $\pi$ from $I(\mathcal{Q})$ to $N$ that satisfies the condition of proportionality restricted to all complete slots in $\mathcal{Q}$:
for any agent $x \in N$ and any index $i \in \{1, \dots, i'_x \}$,
the number of items from $L^x[1:i_x]$ allocated by~$\pi$ to~$x$ is 
at least $i_x/3$; here $i'_x=3(\lfloor (i_x+2)/3 \rfloor)-2$ is the last position in $\Q$ that is contained in a complete slot for agent $x$. 

\smallskip
\noindent
{\bf Minimal obstructions.}
We say that a prefix $\Q$ is a \emph{minimal obstruction}, 
if it is not solvable, but all prefixes strictly contained in $\Q$ are solvable. 
See Figure~\ref{fig-example} for an illustration.
The next lemmas claim some useful observations about minimal obstructions. 

\begin{figure}[t]
\begin{tabular}{lll}
\begin{tabular}[t]{ll}
\multicolumn{2}{l}{Profile $\P$:} \\[4pt]
$a$: & $1, 3, 2, 4, 6, 5, 7$. \\
$b$: & $3, 1, 5, 2, 7, 4, 6$. \\
$c$: & $2, 4, 5, 3, 6, 7, 1$. 
\end{tabular}
&
\hspace{0.4cm}
\begin{tabular}[t]{ll}
\multicolumn{2}{l}{Min. obstruction $\Q$:} \\[4pt]
$a$: & $1, 3, 2, 4$. \\
$b$: & $3, 1, 5, 2$. \\
$c$: & $2, 4, 5, 3$. 
\end{tabular}
&
\hspace{0.4cm}
\begin{tabular}[t]{l}
Graph $G(\Q)$:\\[4pt]
\includegraphics[scale=1]{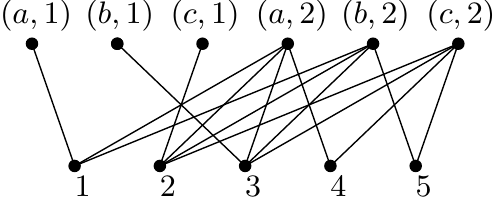}
\end{tabular}
\\
\begin{tabular}[t]{ll}
\multicolumn{2}{l}{$\Q-\{2\}$:} \\[4pt]
$a$: & $\underline{1}, 3, 4$. \\
$b$: & $\underline{3}, 1, 5$. \\
$c$: & $\underline{4}, 5, 3$. 
\end{tabular}
&
\hspace{0.4cm}
\begin{tabular}[t]{ll}
\multicolumn{2}{l}{$\P-\{2\}$:} \\[4pt]
$a$: & $\underline{1}, 3, 4, \underline{6}, 5, 7$. \\
$b$: & $\underline{3}, 1, 5, \underline{7}, 4, 6$. \\
$c$: & $\underline{4}, \underline{5}, 3, 6, 7, 1$. 
\end{tabular}
&
\end{tabular}
\caption{
An example profile $\P$ with item set $I=\{1,2,\dots,7\}$, 
a minimal obstruction $\Q$ of size $(4,4,4)$ in $\P$ 
and its associated graph $G(\Q)$. 
Note that the partial solution $\{2\}$ for $\Q$ is a soluton for $\P$ as well.
We depicted a proportional allocation for $\Q-\{2\}$ and $\P-\{2\}$ by underlining in each agent's preference list the items allocated to her. 
}
\label{fig-example}
\end{figure}

\begin{lemma}
\label{lem-propalloc}
Profile $\P$ admits a proportional allocation if and only if 
the underlying graph $G$ contains a perfect matching. 
Also, in $O(|I|^3)$ time we can find either a proportional allocation for $\P$, or a minimal obstruction $\Q$ in $\P$. 
\end{lemma}
\begin{proof}
We prove this lemma for arbitrary $|N|$. 

First, it is easy to see that any proportional allocation $\pi$ immediately 
yields a perfect matching $M$ for $G$: for each $x \in N$ and each $i \in \{1, \dots, |I|/|N|\}$ (note that $|I|/|N| \in \mathbb{N}$ since $\pi$ is proportional), 
we simply put into $M$ the edge connecting slot $(x,i)$ with the $i$th item $p_{(x,i)}$ received by $x$; naturally, we rank items received by $x$ according to $x$'s preferences. The proportionality of $\pi$ implies that $p_{(x,i)}$ is contained in the top $(i-1)N+1$ items in $L^x$, and thus is indeed eligible for the slot $(x,i)$. 

For the other direction, consider a perfect matching $M$ in $G$. 
Then giving every agent $x$ all the items assigned to 
the slots $\{(x,i) \mid i \in \{1, \dots, |I|/|N|\}$ by $M$
we obtain a proportional allocation $\pi$: 
for each agent $x$ and index $j \in \{1, \dots, |N|\}$, 
our allocation $\pi$ assigns at least $j/|N|$ items to $x$ from $L^x[1:j]$, 
namely the items matched by $M$ to the slots $\{(x,i) \mid 1 \leq i \leq \lceil j/|N| \rceil \}$. Since $(\lceil j/|N| \rceil-1)|N|+1 \leq j$, even 
the last item eligible for $(x, \lceil j/|N| \rceil)$ is 
contained in $L^x[1:j]$, ensuring that $\pi$ is indeed proportional.

We can check whether there exists a proportional allocation for $\P$ 
by finding a maximum matching in a bipartite graph. 
Using the Hopcroft--Karp algorithm~\cite{Hopcroft-Karp-1973},
this takes $O(|I|^{5/2})$ time since $G$ has $2|I|$ vertices. 
To find a minimal obstruction, 
we can use a variant of the classical augmenting path method that
starts from an empty matching, and increases its size by 
finding augmenting paths one by one. 
Namely, at each iteration we pick a starting slot $(x,i)$ 
for which all slots in $\{(x',j)\mid x' \in N, 1 \leq j < i \}$ are already matched,
and search for an augmenting path that starts at $(x,i)$.

Suppose that this algorithm stops at an iteration where 
the starting slot is $(x,i)$, and no augmenting path starts at $(x,i)$ 
for the current matching $M$. 
Let $S_H$ be the set of all slots reachable by an alternating path in $G$ 
from $(x,i)$, and let $I_H$ be the set of all items eligible for any slot in $S_H$. 
It is well known that $S_H$ and $I_H$ violate Hall's condition: $|I_H|<|S_H|$. 
Moreover, the slots in $S_H$ ``induce'' a prefix in the sense that there exists a prefix $\Q$ with $S(\Q)=S_H$. To prove this, 
it suffices to show that if $(y,j) \in S_H$ and $j' \in \{1, \dots, j-1\}$,
then $(y,j') \in S_H$. By our strategy for picking starting slots, we know 
$j' <j \leq i$, implying that $(y,j')$ is matched by $M$. 
Let $q$ be the item assigned to it by $M$; 
note that $q$ is eligible for $(y,j)$ as well. 
To obtain an augmenting path from $(x,i)$ to $(y,j')$, 
we can take any augmenting path from $(x,i)$ to $(y,j)$, 
and append the two-edge path from $(y,j)$ to $(y,j')$ through $q$. 
Hence, there indeed exists a prefix $\Q$ with $S(\Q)=S_H$; 
we pick such a $\Q$ containing only complete slots. 
Using standard arguments from matching theory, 
it is straightforward to check that $\Q$ is a minimal obstruction.

Each iteration can be performed in $O(|I|^2)$ time (e.g., with a BFS), 
and there are at most $|I|$ steps, so the algorithm runs in $O(|I|^3)$ time. 
\qed
\end{proof}

\begin{lemma}
\label{lemma-minobst-shape}
Let $\mathcal{Q}=\mathcal{P}[i_a,i_b,i_c]$ be a prefix of $\mathcal{P}$ 
that is a minimal obstruction. Then $i_a \equiv i_b \equiv i_c \equiv 1 \mod 3$, and either 
\begin{itemize}
\item[(i)] $i_a=i_b=i_c$, or 
\item[(ii)] $i_x=i_y=i_z+3$ for some choice of agents $x$, $y$, and $z$ with $\{x,y,z\}=\{a,b,c\}$. 
\end{itemize}
%(i)~$i_a=i_b=i_c$, or 
%(ii)~$i_x=i_y=i_z+3$ for some choice of agents $x$, $y$, and $z$ with $\{x,y,z\}=\{a,b,c\}$. 
Moreover, if (ii) holds, then  $L^x[1:i_x]$ and $L^y[1:i_y]$ contain exactly the same item set, namely $I(\mathcal{Q})$.
\end{lemma}
\begin{proof}
First, observe that if $i_a \not\equiv 1 \mod 3$,
then the set of complete slots  is the same in $\Q$ as in $\P[i_a-1,i_b,i_c]$, 
contradicting the minimality of $\Q$. Thus, we have $i_a \equiv 1 \mod 3$, and 
we get $i_b \equiv i_c \equiv 1 \mod 3$ similarly. 

Second, let us consider the graph $G(\mathcal{Q})$ underlying our prefix. 
Since Hall's condition fails for the set $S(\mathcal{Q})$ of (complete) slots but,
by minimality, it holds for any proper subset of these slots, we know that 
\begin{equation}
\label{eqn-1}
I(\mathcal{Q})|=|S(\mathcal{Q})|-1= \left\lceil \frac{i_a+2}{3} \right\rceil + \left\lceil \frac{i_b+2}{3} \right\rceil + \left\lceil \frac{i_c+2}{3} \right\rceil -1 = \frac{i_a+i_b+i_c}{3}+1
\end{equation} 
where the last equality follows from the first claim of the lemma.
Let us assume $i_a \geq \max \{ i_b,i_c\}$. 
Note that if neither (i) nor (ii) holds, then from (\ref{eqn-1}) we obtain $i_a+i_b+i_c \leq 3i_a -6$,
yielding $|I(\mathcal{Q})| \leq i_a-1$. However, $L^a[1:i_a]$ contains only items from $I(\mathcal{Q})$, which would imply that some item appears twice in $L^a[1:i_a]$, a contradiction. 

To see the last claim of the lemma, suppose $i_a=i_b=i_c+3$. 
Then (\ref{eqn-1}) implies $|I(\mathcal{Q})|=i_a=i_b$, and hence $L^a[1:i_a]$ (and also $L^b[1:i_b]$) must contain each item in $I(\mathcal{Q})$ exactly once. 
\qed
\end{proof}

Based on Lemma~\ref{lemma-minobst-shape}, we define the \emph{shape} of a minimal obstruction  $\mathcal{Q}$ as either \emph{straight} or \emph{slant}, depending on whether $\mathcal{Q}$ fulfills the conditions (i) or (ii), respectively. More generally, we also say that a prefix has straight or slant shape if it fulfills the respective condition. 
Furthermore, we define the \emph{boundary items} of $\mathcal{Q}$, denoted by $\delta(\mathcal{Q})$, as the set of all items that appear once or twice (but not three times) in~$\mathcal{Q}$.

\begin{lemma}
\label{lemma-boundary}
Let $\mathcal{Q}$ be a prefix of $\mathcal{P}$ that is a minimal obstruction.
Then the boundary of $\mathcal{Q}$ contains at most three items: 
$|\delta(\mathcal{Q})| \leq 3$.
\end{lemma}
\begin{proof}
We make use of Lemma~\ref{lemma-minobst-shape}. 
First, if $\mathcal{Q}$ has a straight shape, so $\mathcal{Q}=\mathcal{P}[i,i,i]$ for some index $i$, then $|S(\Q)|=i+2$. Since $\Q$ is a minimal obstruction, 
we get $|I(\Q)| = i+1$. However, each agent's list within $\Q$ contains exactly $i$ items, yielding that there is exactly one position outside $\Q$ in each agents's list where an item of $I(\Q)$ occurs. Hence, $|\delta(\Q)| \leq 3$ follows in this case.

Second, assume that $\Q$ has a slant shape, say $\mathcal{Q}=\mathcal{P}[i,i,i-3]$ for some index $i$ (the two remaining cases are analogous). Then we have $|S(\Q)|=i+1$, implying $|I(\Q)|=i$. Thus, both $L^a[1:i]$ and  $L^b[1:i]$ contain all the $i$ items in $I(\Q)$, but $L^c[1:i-3]$ misses exactly three items from $I(\Q)$. Hence, there are exactly three occurrences of items listed outside $\Q$, each in the list of agent $c$,  meaning $|\delta(\Q)|=3$. 
\qed
\end{proof}

%We remark that with a more thorough analysis, Lemma~\ref{lemma-boundary} can be sharpened to show that the boundary of a minimum obstruction always contains \emph{exactly} three items.

\subsection{Partial Solutions and Branching Sets}

{\bf Partial solutions.}
For a prefix $\Q$ and a set $U$ of items, we define $\Q-U$ in the natural way: 
by deleting all items of $U$ from the (partial) preference lists of the profile 
(note that the total length of the preference lists constituting the profile may decrease). 
We say that an item set $Y \subseteq I(\Q)$ is 
a \emph{partial solution for $\Q$} if $\mathcal{Q}-Y$ is solvable.
See again Figure~\ref{fig-example} for an example.

Observe that for any item set $Y$ we can check 
whether it is a partial solution for $\Q$
by finding a maximum matching in the corresponding graph (containing all items and complete slots that appear in $\Q-Y$), 
which has at most $2|I|$ vertices. 
%Hence, using the Hopcroft-Karp algorithm~\cite{Hopcroft-Karp-1973}, we can check in $O(|I|^{5/2})$ time
Hence, using the algorithm by Mucha and Sankowski~\cite{mucha-sankowski-max-matching}, we can check for any $Y \subseteq I(\Q)$ whether it is a partial solution for $\Q$
in $O(|I|^{\omega})$ time where $\omega<2.38$ is the exponent of the
best matrix multiplication algorithm.

\smallskip
\noindent
{\bf Branching set.}
To solve PID we will repeatedly apply a branching step: 
whenever we encounter a minimal obstruction $\Q$, we shall consider several possible 
partial solutions for $\Q$, and for each partial solution $Y$ we try to 
find a solution $U$ that contains $Y$. 
To formalize this idea, we say that a family~$\Y$ 
containing partial solutions for a minimal obstruction $\Q$ 
is a \emph{branching set for $\Q$}, 
if there exists a solution $U$ of minimum size for the profile $\P$
such that $U \cap I(\Q) \in \Y$. 
Such a set is exactly what we need to build a search tree algorithm for PID.

Lemma \ref{lemma-numberofdeletions} shows that we never need to delete more than two items 
from any minimal obstruction. 
This will be highly useful for constructing a branching set. 

\begin{lemma}
\label{lemma-numberofdeletions}
Let $\mathcal{Q}$ be a minimal obstruction in a profile $\mathcal{P}$, 
and let $U$ denote an inclusion-wise minimal solution for $\mathcal{P}$. 
Then $|U \cap I(\mathcal{Q})| \leq 2$.
\end{lemma}
\begin{proof}
Let $U_{\mathcal{Q}}:=U \cap I(\mathcal{Q})$, and let us assume $|U_{\mathcal{Q}}| \geq 3$ for contradiction. We are going to select a set $Y$ of three items from $U_{\mathcal{Q}}$ for which we can prove that $U \setminus Y$ is a solution for $\mathcal{P}$, contradicting the minimality of $U$.

We rank the items of $U_{\mathcal{Q}}$ according to the index of the first slot in which they appear in $\mathcal{P}$: we say that an item $u$ \emph{appears at} $i$, if $i$ is the smallest index such that 
$u$ is eligible for a slot $(x,i)$ some $x \in N$. 
%$u$ is contained in $L^x[1:3i-2]$ for some $x \in N$. 
If there exist three items $y_1$, $y_2$, and $y_3$ in $U_{\mathcal{Q}}$ 
appearing strictly earlier than (i.e., at a smaller index) 
than all other items in $U$, then we let $Y=\{y_1,y_2,y_3\}$. 

Otherwise, we apply the following tie braking procedure. 
We select $y_1$ arbitrarily among the earliest appearing items in $U_{\mathcal{Q}}$. Now, let $Y_2$ be the earliest appearing set of items in $U_{\mathcal{Q}} \setminus \{y_1\}$, appearing at some index $i$. We pick an item $y_2$ from $Y_2$ by favoring items eligible to more than one slots from $\{(a,i),(b,i),(c,i)\}$
(note that we use the notion of eligibility based on the original preference lists in $\mathcal{P}$).

To choose an item $y_3$ from the set $Y_3$ of the earliest appearing items in $U_{\mathcal{Q}} \setminus \{y_1,y_2\}$,
we create the profile $\mathcal{P}_3=\mathcal{P}-(U \setminus \{y_1,y_2\})$. 
If there exists a minimal obstruction in $\mathcal{P}_3$ strictly contained in $\mathcal{Q}$, then we fix such a minimal obstruction $\mathcal{Q}_3$, and we choose an item $y_3 \in Y_3$ eligible for a slot of $S(\Q_3)$. Otherwise we choose $y_3$ from $Y_3$ arbitrarily.
Intuitively, we choose $y_3$ so as to overcome the possible obstructions obtained when putting $y_1$ and $y_2$ back into our instance, and our strategy for this is simply to choose an item lying within any such obstruction. Observe that if the minimal obstruction $\mathcal{Q}_3$ exists, then (1) since there is no obstruction strictly contained in $\mathcal{Q}$ in the profile $\mathcal{P}-U_{\Q}$, there must exist some item $u \in U_{\mathcal{Q}} \setminus \{y_1,y_2\}$ that is eligible to some slot in $S(\Q_3)$; and (2) if some $u \in U_{\mathcal{Q}} \setminus \{y_1,y_2\}$ is eligible for a slot in $S(\Q_3)$, then by Lemma~\ref{lemma-minobst-shape} $Y_3$ also contains some item eligible for some slot of $S_3$. Hence, $y_3$ is well-defined. 

Setting $Y=\{y_1,y_2,y_3\}$, we finish our proof by proving that 
$U \setminus Y$ is a solution for $\P$. 
For contradiction, suppose that $\mathcal{R}$ is a minimal obstruction in $\mathcal{P}'=\mathcal{P}-(U \setminus Y)$. 

First, suppose that $\mathcal{R}$ contains all items in $Y$. As $U$ is a solution, the profile  $\mathcal{R}-Y$ is solvable, and hence contains at least as many items as slots. Note that adding the items of $Y$ into the profile $\mathcal{R}-Y$ means adding exactly three new items and at most three new slots (since each agent's list contains at most three more items, resulting in at most one extra slot per agent). Hence, $\mathcal{R}$ has at least as many items as slots, contradicting to the assumption that $\mathcal{R}$ is an obstruction. 

Second, suppose now that $\mathcal{R}$ does not contain all items in $Y$. Then $\mathcal{R}$ is clearly strictly contained\footnote{Seemingly it may be incorrect to say that $\R$ is \emph{contained} in $\Q$ because $\R$ is a prefix of $\P'$ while $\Q$ is a prefix of $\P$; however, recall that the definition of containment only depends on the notion of size.} 
in $\mathcal{Q}$, and by the minimality of $\mathcal{Q}$, we get that $\mathcal{R}$ must contain $y_1$. 

We claim that $y_2$ is contained in $\mathcal{R}$.
For contradiction, assume the opposite. Then $\mathcal{R}$ must contain $y_3$, otherwise $\mathcal{R}$ would be an obstruction in $\mathcal{P}_3$ as well, and we should have chosen an item from $Y_3$ connected to a slot of $S_3$ (as argued above, such an item exists), instead of $y_3$. Hence, we get that $y_3$ is contained in $\mathcal{R}$, but $y_2$ is not. As $y_2$ appears not later than $y_3$, 
this can only happen if $\mathcal{R}$ has slant shape, $y_2$ and $y_3$ appear at the same index $i$, and $y_3$ is eligible to two slots from $\{(a,i),(b,i),(c,i)\}$ 
but $y_2$ is eligible only to the remaining one slot (not contained in $\mathcal{R}$). 
However, this contradicts to our choice for $y_2$, proving our claim.

Hence, we know that both $y_1$ and $y_2$ are contained in $\mathcal{R}$, but not $y_3$. Assume w.l.o.g. that $y_3$ appears at index $j$ in the slot $(c,j)$.
Since $\mathcal{R}$ is an obstruction in $\mathcal{P}_3$ strictly contained in $\mathcal{Q}$, we know that a minimal obstruction $\mathcal{Q}_3$ was found when chosing $y_3$, but $\mathcal{R} \neq \mathcal{Q}_3$. 
Thus, both $\mathcal{R}$ and $\mathcal{Q}_3$ are minimal obstructions fulfilling the condition (ii) of Lemma~\ref{lemma-minobst-shape}, with $\mathcal{R}$ containing the slots $(a,j)$ and $(b,j)$ but not $(c,j)$ (i.e., 
$\mathcal{R}=\mathcal{P}'[3j-2,3j-2,3j-5]$), and $\mathcal{Q}_3$ containing the slot $(c,j)$ and one of $(a,j)$ and $(b,j)$, say $(b,j)$ (i.e., $\mathcal{Q}_3=\mathcal{P}_3[3j-5,3j-2,3j-2]$. Note also that by the last statement of Lemma~\ref{lemma-minobst-shape}, we know $L^a_{I \setminus (U\setminus \{y_1,y_2\})}[1:3j-2]=L^b_{I \setminus (U\setminus \{y_1,y_2\})}[1:3j-2]=L^c_{I \setminus (U\setminus \{y_1,y_2\})}[1:3j-2]$; let $J$ denote this item set. 

Consider the profile $\mathcal{R}'$ in $\mathcal{P}'$ obtained by adding the slot 
$(c,j)$ along with the corresponding items to $\mathcal{R}$, that is, we set
$\mathcal{R}'=\mathcal{P}'[3j-2,3j-2,3j-2]$. From the above arguments we know
that the items contained in $\mathcal{R}'$ are from $J \cup \{y_3\}$, and therefore $\mathcal{R}'$ contains (at most) one more item than $\mathcal{R}$, and exactly one more slot, implying that $\mathcal{R}'$ is an obstruction as well. But we have already seen that no obstruction in $\mathcal{P}'$ can contain $Y$, a contradiction finishing the proof. 
\qed
\end{proof}
 
Lemma~\ref{lemma-numberofdeletions} implies that simply taking all 
partial solutions of $I(\Q)$ of size~1 or~2 yields a branching set for $\Q$. 

\begin{corollary}
\label{cor-branchingset}
For any minimal obstruction $\Q$ in a profile, a branching set $\Y$ for $\Q$ 
of cardinality at most $|I(\Q)|+\binom{|I(\Q)|}{2}=O(|I|^2)$ 
and with $\max_{Y \in \Y} |Y| \leq 2$
can be constructed in polynomial time.
\end{corollary}

\subsection{Domination: Obtaining a Smaller Branching Set}
\label{sec:domination}

To exploit Lemma~\ref{lemma-numberofdeletions} in a more efficient manner,
we will rely on an observation about the equivalence of certain item deletions, 
which can be used to reduce the number of possibilities that we have to explore when encountering a minimal obstruction, i.e., the size of our branching set. 
To this end, we need some additional notation. Given a prefix $\mathcal{Q}=\mathcal{P}[i_a,i_b,i_c]$, we define its \emph{tail} as the set $T(\mathcal{Q})$ of items  as follows, depending on the shape of $\mathcal{Q}$. 
\begin{itemize}
\item If $\mathcal{Q}$ has straight shape, then $T(\mathcal{Q})$ contains the last three items contained in $\mathcal{Q}$ for each agent, that is, all items in $L^a[i_a-2:i_a]$, $L^b[i_b-2:i_b]$, and $L^c[i_c-2:i_c]$.
%\item If $\mathcal{Q}$ has slant shape with $i_c=i_a-3=i_b-3$ (the other two cases are  analogous) then $T(\mathcal{Q})$ contains the last six items in $\mathcal{Q}$ listed by agents $a$ and $b$, that is, all items in $L^a[i_a-5:i_a]$ and $L^b[i_b-5:i_b]$.
\item If $\mathcal{Q}$ has slant shape with $i_z=i_x-3=i_y-3$ for some choice of agents $x,y$, and~$z$ with $\{x,y,z\}= \{a,b,c\}$, then $T(\mathcal{Q})$ contains the last six items in $\mathcal{Q}$ listed by agents $x$ and $y$, that is, all items in $L^x[i_x-5:i_x]$ and $L^y[i_y-5:i_y]$.

\end{itemize}
Let us state the main property of the tail 
which motivates its definition. 

\begin{lemma}
\label{lemma-tail-prop}
Suppose $\Q$ is a minimum obstruction in $\P$, and $\R$ is a prefix of $\P$
intersecting $\Q$ such that $\R-X$ is a minimum obstruction 
for some item set $X$ with $|X \cap I(\Q)|\leq 2$.
Then any item that occurs more times in $\mathcal{Q}$ than in $\mathcal{R}$ 
must be contained in the tail of $\mathcal{Q}$. 
\end{lemma}
\begin{proof}
To prove the claim, it will be useful to keep in mind that
both $\mathcal{Q}$ and $\mathcal{R}':=\mathcal{R}-X$ have the shape of a minimal obstruction, so by Lemma~\ref{lemma-minobst-shape}, they either have a straight or a slant shape.
First suppose that $\mathcal{Q}$ has a straight shape and thus $\mathcal{Q}=\mathcal{P}[i,i,i]$ for some index $i$;  let $\mathcal{R}'=\mathcal{P}'[i_a,i_b,i_c]$ where $\mathcal{P}'=\mathcal{P}-X$. 
Since $\mathcal{R}$ is not contained in $\mathcal{Q}$ and $|X \cap I(\Q)| \leq 2$, 
we know that 
$i' \geq i$ for some $i' \in \{i_a,i_b,i_c\}$. On the other hand,  $\mathcal{Q}$ is not contained in $\mathcal{R}$, so $i'' \leq i-3$ for some $i'' \in \{i_a,i_b,i_c\}$. 
Therefore, Lemma~\ref{lemma-minobst-shape} implies that $\mathcal{Q}$ has a slant shape, with $\max\{i_a,i_b,i_c\}=i$ and $\min\{i_a,i_b,i_c\}=i-3$.
Second, suppose now that $\mathcal{Q}$ has a slant shape with $\mathcal{Q}=\mathcal{P}[i,i,i-3]$ for some index $i$ (the remaining two cases are analogous); let again $\mathcal{R}'=\mathcal{P}'[i_a,i_b,i_c]$. Note that $i_c \geq i-3$, as otherwise $\mathcal{R}$ would be contained in $\mathcal{Q}$. From this follows that $\min\{i_a,i_b\} \geq i-6$. It remains to observe that the definition of the tail $T(\mathcal{Q})$ ensures that our claim holds in both cases. 
\qed
\end{proof}

Next, we give a condition that guarantees that some partial solution %$Y'$ 
for a minimum obstruction $\Q$ is ``not worse'' than some other. %partial solution $Y$. 
Given two sets of items $Y, Y' \subseteq I(\mathcal{Q})$, we say that $Y'$ \emph{dominates} $Y$ with respect to the prefix $\mathcal{Q}$, if 
%(1)~$|Y|=|Y'|$, and
%(2)~$Y'$ only contains an item from the boundary or the tail of $\Q$ 
%if $Y$ also contains that item, i.e., 
%$Y' \cap (\delta(\mathcal{Q}) \cup T(\mathcal{Q})) \subseteq 
%Y \cap (\delta(\mathcal{Q}) \cup T(\mathcal{Q}))$. 
\begin{itemize}
\item[(1)] $|Y|=|Y'|$, 
\item[(2)] $Y'$ only contains an item from the boundary or the tail of $\Q$ 
if $Y$ also contains that item, i.e., 
$Y' \cap (\delta(\mathcal{Q}) \cup T(\mathcal{Q})) \subseteq 
Y \cap (\delta(\mathcal{Q}) \cup T(\mathcal{Q}))$. 
\end{itemize} 

\begin{lemma}
\label{lemma-dominating-deletion}
If $U$ is an inclusion-wise minimal solution for the profile $\P$,  
$\mathcal{Q}$ is a minimal obstruction in $\mathcal{P}$, 
$Y= U \cap I(\mathcal{Q})$
and $Y' \subseteq I(\mathcal{Q})$ is a partial solution for $\Q$ that dominates $Y$, 
then $U \setminus Y \cup Y'$ is a solution for $\mathcal{P}$.
\end{lemma}
\begin{proof}
Let $U'=U \setminus Y \cup Y'$.
%Clearly, $|Y|=|Y'|$ implies $|U'|=|U|$, so we only need to prove that $\mathcal{P}-U'$ is solvable. 
For contradiction, suppose that $\mathcal{R}'$ is a minimal obstruction in $\mathcal{P}-U'$. Let us consider a prefix $\mathcal{R}$ of $\mathcal{P}$ for which $\mathcal{R}-U'=\mathcal{R}'$. Depending on the relation between $\mathcal{R}$ and $\mathcal{Q}$, we distinguish three cases. 

First, suppose $\mathcal{Q}$ contains $\mathcal{R}$. Then $\mathcal{Q}-U'=\mathcal{Q}-Y'$ contains $\mathcal{R}-U'=\mathcal{R}'$. However, 
as $Y'$ is a partial solution for $\Q$, 
there is no obstruction in $\mathcal{Q}-Y'$, a contradiction.

Second, assume that $\mathcal{R}$ contains $\mathcal{Q}$. We are going to prove that in this case $\mathcal{R}-U$ is not solvable, contradicting the assumption that $U$ is a solution. 
Observe that $Y \cup Y' \subseteq I(\mathcal{Q}) \subseteq I(\mathcal{R})$, 
so by $U \triangle U'= Y \triangle Y'$ we know that 
% $|I(\mathcal{R}-U')|=|I(\mathcal{R})|-|Y'|=|I(\mathcal{R})|-|Y|=|I(\mathcal{R}-Y)|$.
$|I(\mathcal{R}-U')|=|I(\mathcal{R}-U)|$. Moreover, from $Y' \cap \delta(\mathcal{Q}) \subseteq Y \cap \delta(\mathcal{Q})$ 
it follows that deleting $U$ removes at most as many items  \emph{from each agent's list} in $\mathcal{R}$ as deleting $U'$: indeed, any item $u' \in U'$ either is contained in $U$ as well, or appears three times in $\mathcal{Q}$ (and hence in $\mathcal{R}$). 
This implies that $|S(\mathcal{R}-U')| \leq |S(\mathcal{R}-U)|$. Using that $\mathcal{R}'=\mathcal{R}-U'$ is an obstruction, we get that 
$|S(\mathcal{R}-U)| \geq |S(\mathcal{R}-U')| > |I(\mathcal{R}-U')|=|I(\mathcal{R}-U)|$ and thus $\R-U$ is an obstruction too, a contradiction. 

Third, suppose that neither $\mathcal{Q}$ nor $\mathcal{R}$ contains the other.
%%% Proof without the lemma about the tail:
% We claim that any item that occurs more times in $\mathcal{Q}$ than in $\mathcal{R}$ must be contained in the tail of $\mathcal{Q}$. To see this, it will be useful to keep in mind that both $\mathcal{Q}$ and $\mathcal{R}'=\mathcal{R}-U'$ have the shape of a minimal obstruction, so by Lemma~\ref{lemma-minobst-shape}, they either have a straight or a slant shape. First suppose that $\mathcal{Q}$ has a straight shape and thus $\mathcal{Q}=\mathcal{P}[i,i,i]$ for some index $i$;  let $\mathcal{R}'=\mathcal{P}'[i_a,i_b,i_c]$ where $\mathcal{P}'=\mathcal{P}-U'$. Since $\mathcal{R}$ is not contained in $\mathcal{Q}$ and $|Y|'=|Y| \leq 2$ by Lemma~\ref{lemma-numberofdeletions}, we know that $i' \geq i$ for some $i' \in \{i_a,i_b,i_c\}$. On the other hand,  $\mathcal{Q}$ is not contained in $\mathcal{R}$, so $i'' \leq i-3$ for some $i'' \in \{i_a,i_b,i_c\}$. Therefore, Lemma~\ref{lemma-minobst-shape} implies that $\mathcal{Q}$ has a slant shape, with $\max\{i_a,i_b,i_c\}=i$ and $\min\{i_a,i_b,i_c\}=i-3$. Second, suppose now that $\mathcal{Q}$ has a slant shape with $\mathcal{Q}=\mathcal{P}[i,i,i-3]$ for some index $i$ (the remaining two cases are analogous); let again $\mathcal{R}'=\mathcal{P}'[i_a,i_b,i_c]$. Note that $i_c \geq i-3$, as otherwise $\mathcal{R}$ would be contained in $\mathcal{Q}$. From this follows that $\min\{i_a,i_b\} \geq i-6$. It remains to observe that the definition of the tail $T(\mathcal{Q})$ ensures that our claim holds in both cases. 
We apply Lemma~\ref{lemma-tail-prop} with $U'$ taking the role of the item set $X$ 
(and $\Q$ and $\R$ keeping their meaning); 
note that the conditions of the lemma are fulfilled, 
since $|U' \cap I(\Q)| = |Y|'=|Y| \leq 2$ by Lemma~\ref{lemma-numberofdeletions}.
From the lemma we get that any item that occurs more times in $\mathcal{Q}$ than in $\mathcal{R}$ is contained in the tail $T(\mathcal{Q})$.

Now, since $Y'$ dominates $Y$, we know that $Y' \cap (\delta(\mathcal{Q}) \cup T(\mathcal{Q})) \subseteq Y \cap (\delta(\mathcal{Q}) \cup T(\mathcal{Q}))$. 
Hence, deleting $U$ removes at most as many items  \emph{from each agent's list} in $\mathcal{R}$ as deleting $U'$: indeed, any item $u' \in U'$ either is contained in $U$ as well, or $u' \notin \delta(\mathcal{Q}) \cup T(\mathcal{Q})$, which means that $u'$ appears (three times) in $\mathcal{Q}$ and, using our claim, also in $\mathcal{R}$. 
This observation implies, on the one hand, $|S(\mathcal{R}-U')| \leq |S(\mathcal{R}-U)|$, and on the other hand, $|I(\mathcal{R}-U)| \geq |I(\mathcal{R}-U')|$. Using that
$\mathcal{R}-U'$ is a minimal obstruction, we get 
$|S(\mathcal{R}-U)| \geq |S(\mathcal{R}-U')|> |I(\mathcal{R}-U)| \geq |I(\mathcal{R}-U')|$, yielding that $\mathcal{R}-U$ cannot be solvable, contradicting our assumption that $U$ is a solution for $\P$. 
\qed
\end{proof}

Lemma~\ref{lemma-dominating-deletion} means that if a branching set $\Y$ 
contains two different partial solutions $Y$ and $Y'$ for a minimum obstruction
such that $Y'$ dominates $Y$, 
then removing $Y$ from $\Y$ still results in a branching set. 
Using this idea, we can construct a branching set of constant size.

\begin{lemma}
\label{lemma-small-branchingset}
There is a polynomial-time algorithm that, 
given a minimal obstruction $\Q$ in the profile $\P$,
produces a branching set $\mathcal{Y}$  
with $\max_{Y \in \Y} |Y|  \leq 2$
and $|\Y|=O(1)$. 
\end{lemma}

\begin{proof}
First observe that for any two item sets $Y$ and $Y'$ in $\Q$, we can decide 
whether $Y$ dominates $Y'$ in $O(\min(|Y|,|Y'|))$ time. 
Hence, we can simply start from the branching set $\Y$ guaranteed by Corollary~\ref{cor-branchingset}, and check for each $Y \in \Y$ whether there exists some $Y' \in \Y$ that dominates $Y$; if so, then we remove $Y$. 
By Lemma \ref{lemma-dominating-deletion}, at the end of this process 
the set family $\Y$ obtained is a branching set.

We claim that $\Y$ has constant size. To see this, observe that if 
$Y_1$ and $Y_2$ are both in $\Y$ and have the same size, then both $Y_1 \setminus Y_2$ and $Y_2 \setminus Y_1$ contain an element from $T(\Q) \cup \delta(\Q)$. 
Thus, we can bound $|\Y|$ using the pigeon-hole principle:
first, $\Y$ may contain at most $|T(\Q) \cup \delta(\Q)|$ 
partial solutions of size 1, and second, 
it may contain at most $\binom{|T(\Q) \cup \delta(\Q)|}{2}$ 
partial solutions of size 2. Recall that $|T(\Q)| \leq 9$ by definition, and 
we also have $|\delta(\Q)| \leq 3$ by Lemma~\ref{lemma-boundary}, proving our claim. 
\qed
\end{proof}

\subsection{Polynomial-Time Algorithm for PID for Three Agents}

Let us now present our algorithm for solving PID on our profile $\mathcal{P}=(N,I,L)$. 

We are going to build the desired solution step-by-step,
iteratively extending an already found partial solution.
Namely, we propose an algorithm $\MinDel(\T,U)$ that, 
given a prefix $\T$ of $\P$ and a partial solution $U$ for $\T$,
returns a solution $S$ for $\P$ for which $S \cap I(\T)=U$, 
and has minimum size among all such solutions.
We refer to the set $S \setminus U$ as an \emph{extension} for $(\T,U)$; 
note that an extension for $(\T,U)$ only contains items from $I \setminus I(\T)$.
We will refer to the set of items in $I(\T) \setminus U$ as 
\emph{forbidden} w.r.t. $(\T,U)$. 

\smallskip
\noindent
{\bf Branching set with forbidden items.}
To address the problem of finding an extension for $(\T,U)$, we 
modify the notion of a branching set accordingly.
Given a minimal obstruction $\Q$ and a set $F \subseteq I(\Q)$ of 
items, we say that a family $\Y$ 
of partial solutions for $\Q$ is 
a \emph{branching set for $\Q$ forbidding $F$}, if the following holds: 
either there exists a solution $U$ for the profile $\P$ that is 
disjoint from $F$ and has minimum size among all such solutions, and moreover, fulfills $U \cap I(Q) \in \Y$, 
or $\P$ does not admit any solution disjoint from $F$.

\begin{lemma}
\label{lemma-small-forbidding-branchingset}
There is a polynomial-time algorithm that, 
given a minimal obstrucion $\Q$ in a profile and a set $F \subseteq I(\Q)$
of forbidden items, produces a branching set $\Y$ forbidding $F$ 
with $\max_{Y \in \Y} |Y| \leq 2$ and $|\Y|=O(1)$.
\end{lemma}

\begin{proof}
The algorithm given in Lemma~\ref{lemma-small-branchingset} 
can be adapted in a straightforward fashion to take forbidden items into account: 
it suffices to simply discard in the first place any subset 
$Y \subseteq I(\Q)$ that is not disjoint from $F$. 
It is easy to verify that this modification indeed yields an algorithm
as desired.
\qed
\end{proof}

\noindent
{\bf Equivalent partial solutions.}
We will describe $\MinDel$ as a recursive algorithm, but in order to ensure
that it runs in polynomial time, we need to apply dynamic programming. 
For this, we need a notion of equivalence: we say that two partial solutions $U_1$ and $U_2$ for $\T$ are \emph{equivalent}  if 
\begin{enumerate}
\item $|U_1|=|U_2|$, and
\item $(\T,U_1)$ and $(\T,U_2)$ admit the same extensions.
\end{enumerate}
%(1)~$|U_1|=|U_2|$, and
%(2)~$(\T,U_1)$ and $(\T,U_2)$ admit the same extensions.

Ideally, whenever we perform a call to $\MinDel$ with a given input $(\T,U)$, 
we would like to first check whether an equivalent call has already been performed, i.e., whether $\MinDel$ has been called with an input $(\T,U')$ for which $U$ and $U'$ are equivalent. However, the above definition of equivalence is computationally hard to handle: there is no easy way to check whether two partial solutions admit the same extensions or not. 
To overcome this difficulty, we will use a stronger condition that implies equivalence.

\smallskip
\noindent
{\bf Deficiency and strong equivalence.}
Consider a solvable prefix $\Q$ of $\P$. 
We let the \emph{deficiency} of $\Q$, 
denoted by $\mydef(\Q)$, be the value $|S(\Q)|-|I(\Q)|$. 
Note that due to possibly incomplete slots in $\Q$, the deficiency of $\Q$ may be positive even though $\Q$ is solvable. However, if $\Q$ contains only complete slots, then its solvability implies $\mydef(\Q) \leq 0$. 
We define the \emph{deficiency pattern} of $\Q$ as the set of all 
triples $$(\mysize(\Q \cap \R),\mydef(\Q \cap \R),I(\Q \cap \R) \cap \delta(Q))$$ 
where $\R$ can be any prefix with a straight or a slant shape that intersects $\Q$.
Roughly speaking, the deficiency pattern captures all the information about $\Q$ 
that is relevant for determining whether a given prefix intersecting $\Q$ 
is a minimal obstruction or not.

Now, we call the partial solutions $U_1$ and $U_2$ for $\T$
\emph{strongly equivalent}, if 
\begin{enumerate}
\item $|U_1|=|U_2|$,
\item $U_1 \cap \delta(\T)=U_2 \cap \delta(\T)$, and
\item $\T-U_1$ and $\T-U_2$ have the same deficiency pattern. 
\end{enumerate} 
As the name suggests, strong equivalence is a sufficient condition for equivalence. 

\begin{lemma}
\label{lemma-equivalent-calls}
If  $U_1$ and $U_2$ are strongly equivalent partial solutions for $\T$, 
then they are equivalent as well.  
\end{lemma}
\begin{proof}
Suppose that $W$ is an extension for $(\T,U_1)$, we need to prove that it is an extension for $(\T,U_2)$ as well. 
Clearly, we have $W \subseteq I \setminus I(\T)$, so let us show that 
$\P-(U_2 \cup W)$ is solvable. 

Suppose for contradiction that $\Q_2$ is a minimum obstruction in $\P-(U_2 \cup W)$. 
Let us consider the prefix $\Q_1$ of $\P-(U_1 \cup W)$ that has the 
same size as $\Q_2$. 
In the remainder of the proof, we argue that $\Q_1$ is not solvable, contradicting to the assumption that $W$ is an extension for $(\T,U_1)$ 
(and thus $U_1 \cup W$ is a solution for $\P$). 
Note that $\Q_1$ and $\Q_2$ clearly have the same slots, all of them complete. 

Recall that $U_2$ is a partial solution for $\T$, so $\T-U_2$ is solvable.
Hence, $\Q_2$ cannot be contained in $\T-U_2$. 

First, let us assume that $\Q_2$ contains $\T-U_2$. 
In this case, $|U_1|=|U_2|$ and $U_1 \cap \delta(\T)=U_2 \cap \delta(\T)$
together immediately imply that $\Q_1$ and $\Q_2$ contain the same number of items: 
$|I(\Q_1)|=|I(\Q_2)|$. 
Hence, we get $|I(\Q_1)|=|I(\Q_2)|<|S(\Q_2)|=|S(\Q_2)|$, proving our claim. 

Second, let us assume now that $\Q_2$ and $T-U_2$ are intersecting, and let their intersection be $\T_2^{\cap}$. 
Since $\Q_2$ is a minimum obstruction and thus has a straight or a slant shape, 
we know that $(\mysize(\T_2^{\cap}),\mydef(\T_2^{\cap}),I(\T_2^{\cap}) \cap \delta(\T-U_2))$ is contained in the deficiency pattern of $\T-U_2$. 
By the third condition of equivalence, the same triple must also be present in the deficiency pattern of $\T-U_1$. 
Now, let $\T_1^{\cap}$ be the prefix in $\T-U_1$ that has the same size as $\T_2^{\cap}$. Then $\T_1^{\cap}$ must have the same deficiency as $\T_2^{\cap}$, 
implying 
$|I(\T_1^{\cap})|=|I(\T_2^{\cap})|$. 
Moreover, we also get 
$I(\T_2^{\cap}) \cap \delta(\T-U_2)=I(\T_1^{\cap}) \cap \delta(\T-U_1)$. 
Recall that $U_1 \cap \delta(\T) = U_2 \cap \delta(\T)$
is guaranteed by the second condition of strong equivalence. 
Hence, adding the items contained in $\Q_2$ but not in $\T$ increases the size 
of $I(\T_2^{\cap})$ exactly as adding the items contained in $\Q_1$ but not in $\T$ increases the size of $I(\T_1^{\cap})$. Therefore, we can conclude that $|I(\Q_1)|=|I(\Q_2)|$, which again implies that $\Q_1$ is not solvable. 
\qed
\end{proof}

%\medskip

Now, we are ready to describe the $\MinDel$ algorithm in detail.
Let $(\T,U)$ be the input for $\MinDel$. 
Throughout the run of the algorithm, we will store all inputs with which 
$\MinDel$ has been computed in a 
table $\TableComp$, 
keeping track of the corresponding solutions for $\P$ as well. 
Initially, $\TableComp$ is empty.

{\bf Step 0: Check for strongly equivalent inputs.}
For each $(\T,U')$ in $\TableComp$, check whether $U'$ and $U$ are strongly equivalent, and if so, return $\MinDel(\T,U')$.
 
{\bf Step 1: Check for trivial solution.} 
Check if $\P-U$ is solvable. 
If so, then store the entry $(\T,U)$ together with the solution $U$ in $\TableComp$, and return $U$. 

{\bf Step 2: Find a minimal obstruction.} 
Find a minimal obstruction $\Q$ in $\P-U$; 
recall that $\P-U$ is not solvable in this step.
Let $\T'$ be the prefix of $\P$ for which $\T'-U=\Q$.

{\bf Step 3: Compute a branching set.} 
Using Lemma~\ref{lemma-small-forbidding-branchingset},
determine a branching set $\mathcal{Y}$ for $\Q$ 
forbidding $I(\T) \setminus U$. 
If $\Y=\emptyset$, then stop and reject. 
%Recall that each partial solution $Y \in \mathcal{Y}$ for $\Q$ 
%contains at most two items, and $|\Y| = O(1)$.

{\bf Step 4: Branch.} For each $Y \subseteq \mathcal{Y}$,  
compute $S_Y:=\MinDel(\T',U \cup Y)$. 

{\bf Step 5: Find a smallest solution.} 
Compute a set $S_{Y^\star}$ for which $|S_{Y^\star}|=\min_{Y \in \Y} |S_Y|$. 
Store the entry $(\T,U)$ together with the solution $S_{Y^\star}$ in $\TableComp$, 
and return $S_{Y^\star}$.

\begin{lemma}
\label{lemma-soundness}
Algorithm $\MinDel$ is correct, i.e., 
for any prefix $\T$ of $\P$ and any partial solution $U$ for $\T$, 
$\MinDel(\T,U)$ returns a solution $S$ for $\P$ with $S \cap I(\T)=U$, 
having minimum size among all such solutions (if existent).
\end{lemma}
\begin{proof}
Observe that it suffices to prove the claim for those cases 
when algorithm $\MinDel$ does not return a solution in Step 0: 
its correctness in the remaining cases (so when a solution contained 
in $\TableComp$ for a strongly equivalent input is found and returned in Step 0)
follows from Lemma~\ref{lemma-equivalent-calls}. 

We are going to prove the lemma by induction on $|I \setminus U|$. 
Clearly, if $|I \setminus U|=0$, then $\P-U$ is an empty instance, 
and hence is trivially solvable. 
Assume now that $I \setminus U \neq \emptyset$, and that 
$\MinDel$ returns a correct output for any input $(\T_0,U_0)$ 
with $|I \setminus U_0|<|I \setminus U|$. 

First, if the algorithm returns $U$ in Step 1, then this is clearly correct. 
Also, if it stops and rejects in Step 3 because it finds that the branching set $\Y$ forbidding $I(\T) \setminus U$ for the minimal obstruction $\Q$ is empty, 
then this, by the definition of a branching set (forbidding $I(\T) \setminus U$) 
and by the soundness of the algorithm of Lemma~\ref{lemma-small-forbidding-branchingset}, means that there is no solution $S$ for $\P-U$ disjoint from $I(\T) \setminus U$ holds. But then we also know that there is no solution $S$ for $\P$ for which $S \cap I(\T)=U$ holds. 
Hence this step is correct as well. 

Therefore, we can assume that the algorithm's output is $S_{Y^\star}$ for some 
$Y^\star \in \Y$, where $S_{Y^\star}=\MinDel(\T', U \cup Y^\star)$ and 
$\T'$ is the profile in $\P$ for which $\T'-U=\Q$. 
As $|I \setminus (U \cup Y)| < |I \setminus U|$ for any $Y \in \Y$, 
the induction hypothesis implies that $\MinDel$ 
runs correctly on all inputs $(\T', U \cup Y)$, $Y \in \Y$.
Hence, $S_{Y^\star}$ is a solution for $\P$ for which 
$I(\T') \cap S_{Y^\star} = U \cup Y^\star$, and since $Y^\star$ is contained 
in a branching set forbidding $I(\T) \setminus U$, we get 
$I(\T) \cap  S_{Y^\star} = U$. 

It remains to argue that if $S$ is a solution for $\P$ with 
$I(\T) \cap S= U$, then $|S_{Y^\star}| \leq S$. 
Observe that $S \setminus U$ is a solution for $\P-U$ forbidding $I(\T) \setminus U$. 
By the definition of a branching set forbidding $I(\T) \setminus U$
and the correctness of Lemma~\ref{lemma-small-forbidding-branchingset}, 
we know that there must exist a solution $S'$ for $\P-U$ disjoint from $I(\T) \setminus U$ and having size $|S'| \leq |S \setminus U|$ 
for which $S' \cap I(\Q) \in \Y$. 
Let $Y'=S' \cap I(\Q)$. 

Note that $S' \cup U$ is a solution for $\P$ with $|S' \cup U| \leq |S|$, and 
clearly, $(S' \cup U) \cap I(\T')= U \cup Y'$.
Using again the induction hypothesis, we
get that $\MinDel(\T',U \cup Y')$ returns a solution $S_{Y'}$ for $\P$ for which 
$S_{Y'} \cap I(\T')= U \cup Y'$ and has minimum size among all such solutions, 
so in particular, $|S_{Y'}| \leq |S' \cup U| \leq |S|$. 
Thus, we get that the output of 
$\MinDel(\T,U)$ (that is, $S_{Y^\star}$) has size at most $|S|$. 
This proves our claim. 
Therefore, we get that if $\MinDel$ returns an output in Step 5, then this 
output is correct. 
\qed
\end{proof}

Lemma~\ref{lemma-soundness} immediately gives us an algorithm to solve PID.
Let $\T_{\emptyset}$ denote the empty prefix of our input profile $\P$, 
i.e. $\P[0,0,0]$; then $\MinDel(\T_{\emptyset},\emptyset)$ returns 
a solution $S$ for $\P$ of minimum size; 
we only have to compare $|S|$ with the desired solution size $k$.  

The next lemma states that $\MinDel$ gets called polynomially many times.

\begin{lemma}
\label{lemma-tablesize-bounded}
Throughout the run of algorithm $\MinDel$ initally called with input $(\T_{\emptyset},\emptyset)$, the table $\TableComp$ contains $O(|I|^{7})$ entries. 
\end{lemma}

\begin{proof}
Let us consider table $\TableComp$ at a given moment 
during the course of algorithm $\MinDel$, initially called 
with the input $(\T_{\emptyset},\emptyset)$ (and having possibly performed 
several recursive calls since then). 
Let us fix a prefix $\T$. 
We are going to give an upper bound on the maximum size 
of the family $\U_{\T}$ of partial solutions~$U$ for $\T$ for which 
$\TableComp$ contains the entry $(\T,U)$. 

By Step 0 of algorithm $\MinDel$, no two sets in $\U_{\T}$ are strongly equivalent.
Recall that if $U_1$ and $U_2$, both in $\U_{\T}$, are not strongly
equivalent, then either $|U_1| \neq |U_2|$, 
or $\delta(\T) \cap U_1 \neq \delta(\T) \cap U_2$, 
or $\T-U_1$ and $\T-U_2$ have different deficiency patterns. 
Let us partition the sets in $\U_{\T}$ into \emph{groups}: 
we put $U_1$ and $U_2$ in the same group, 
if $|U_1|=|U_2|$ and $\delta(\T) \cap U_1 = \delta(\T) \cap U_2$. 

Examining Steps 2--4 of algorithm $\MinDel$, we can observe that 
if $U \neq \emptyset$, then for some $Y_U \subseteq U$ of size 1 or 2, 
the prefix $\T-(U \setminus Y_U)$ is a minimal obstruction~$\Q_U$. 
Since removing items from a prefix cannot increase the size of its boundary, 
Lemma~\ref{lemma-boundary} implies that the boundary of $\T-U$ 
contains at most 3 items. 
We get $|\delta(\T) \setminus U| \leq |\delta(\T-U)| \leq 3$, 
from which it follows that $\delta(\T) \cap U$ is a subset of $\delta(\T)$ 
of size at least $|\delta(\T)|-3$. 
Therefore, the number of different values that $\delta(\T) \cap U$ 
can take is $O(|I|^3)$. Since any $U \in \U_{\T}$ has size at most $|I|$, 
we get that there are $O(|I|^4)$ groups in $\U_{\T}$. 
Let us fix some group  $\U_g$ of $\U_{\T}$.
We are going to show that 
the number of different deficiency patterns for $\T-U$ 
where $U \in \U_g$ is constant. 

Recall that the deficiency pattern of $\T-U$ contains triples of the form
$$(\mysize(\R^{\cap}),\mydef(\R^{\cap}),I(\R^{\cap}) \cap \delta(\T-U))$$ 
%$(\mysize(\R^{\cap}),\mydef(\R^{\cap}),I(\R^{\cap}) \cap \delta(\T-U))$,
where $\R^{\cap}$ is the intersection of $\T-U$ 
and some prefix $\R$ of $\P-U$
with a slant or a straight shape. 

First observe that by the definition of a group, 
$\mysize(\T-U_1)=\mysize(\T-U_2)$ holds for any $U_1, U_2 \in \U_g$. 
Let us fix an arbitrary $U \in \U_g$.
Since $\T-U$ can be obtained by deleting 1 or 2 items 
from a minimal obstruction,  
Lemma~\ref{lemma-minobst-shape} implies that 
there can only be a constant number of prefixes $\R$ of $\P-U$ 
which intersect $\T-U$ and have a slant or a straight shape;
in fact, it is not hard to check that 
the number of such prefixes $R$ is at most 5 for any given $\T-U$. 
%TODO: figure! 
Therefore, the number of values taken by the first coordinate 
$\mysize(\R^{\cap})$ of any triple in the deficiency pattern of $\T-U$ is constant.
Since $\T-U$ has the same size for any $U \in \U_g$, we also get that 
these values coincide for any $U \in U_g$. 
Hence, we obtain that (A) the total number of values the first coordinate of any triple in the deficiency pattern of $\T-U$ for any $U \in \U_g$ can take is constant.

Let $\R_{\cap}$ be the intersection of $\T-U$ and 
some prefix of straight or slant shape. 
By definition, $\R_{\cap}$ is contained in $\Q_U$. 
By $|Y_U| \leq 2$, there are only a constant number of positions which are contained in $\Q_U$ but not in $\R_{\cap}$. From this both
$||I(\R_{\cap})| - |I(\Q_U)||=O(1)$ and $||S(\R_{\cap})| - |S(\Q_U)||=O(1)$ follow.
As $\Q_U$ is a minimal obstruction, we also have $|I(\Q_U)|=|S(\Q_U)|-1$, 
implying that (B) the deficiency 
$\mydef(\R_{\cap})= |S(\R_{\cap})|-|I(\R_{\cap})|$ can only take a constant 
number of values too; note that we have an upper bound on $|\mydef(\R_{\cap})|$ that holds for any $U \in \U_g$. 
Considering that $I(\R_{\cap}) \cap \delta(\T-U)$ is the subset of $\delta(\T-U)$, 
and we also know $|\delta(\T-U)| \leq 3$, we obtain that 
(C) the set $I(\R_{\cap}) \cap \delta(\T-U)$ can take at most $2^3$ values 
(again, for all $U \in U_g$). 

Putting together the observations (A), (B), and (C), it follows that 
the number of different deficiency patterns of $\T-U$
taken over all $U \in \U_{g}$ is constant. This implies $|\U_{\T}|=O(|I|^4)$.
Since there are $O(|I|^3)$ prefixes $\T$ of $\P$, we arrive at the conclusion 
that the maximum number of entries in $\TableComp$ is $O(|I|^7)$.
\qed
\end{proof}

\begin{theorem}%[$\star$]
\label{thm-3agents-poly}
\textsc{Proportional Item Deletions} for three agents can be solved in time~$O(|I|^{9+{\omega}})$ where $\omega<2.38$ is the exponent of the
best matrix multiplication algorithm.
\end{theorem}

\begin{proof}
By Lemma~\ref{lemma-soundness}, we know that algorithm $\MinDel(\T_{\emptyset},\emptyset)$ returns a solution for $\P$ of minimum size, solving PID. 
We can use Lemma~\ref{lemma-tablesize-bounded} to bound the running time of 
$\MinDel(\T_{\emptyset},\emptyset)$:
since~$\TableComp$ contains~$O(|I|^7)$ entries,
we know that the number of recursive calls to~$\MinDel$ is also~$O(|I|^7)$.
%since there can only be a polynomial number of 
%entries in $\TableComp$, this means that the number of recursive calls to 
%algorithm $\MinDel$ can also be upper-bounded by a polynomial. 
It remains to give a bound on the time necessary for the computations performed 
by $\MinDel$, when not counting the computations performed in recursive calls.
Clearly, Step 0 takes $O(1)$ time. 
Steps 1 and 2 can be accomplished in $O(|I|^3)$ time, 
as described in Lemma~\ref{lem-propalloc}.
Using Lemma~\ref{lemma-small-forbidding-branchingset}, Step 3 can be performed in
$O(|I|^{2+{\omega}})$ time. Since the cardinality of the branching set found in 
Step 3 is constant, Steps 4 and 5 can be performed in linear time. 
%Thus, each call to algorithm $\MinDel$ requires polynomial time for its computations,
%and since $\MinDel$ is called only a polynomial number of times when solving 
%our PID instance, the theorem follows.
This gives us an upper bound of~$O(|I|^{9+{\omega}})$ on the total running time.
\qed
\end{proof}

We remark that in order to obtain Theorem~\ref{thm-3agents-poly}, it is not crucial to compute a branching set of constant size in Step 3: a polynomial running time would still follow even if we used a branching set of quadratic size. Thus, for our purposes, it would be sufficient to use an extension of Corollary~\ref{cor-branchingset} that takes forbidden items into account (an analog of Lemma~\ref{lemma-small-forbidding-branchingset}) in Step 3. Therefore, the ideas of Section~\ref{sec:domination} -- the notion of domination between partial solutions, leading to Lemma~\ref{lemma-small-branchingset} -- can be thought of as a speed-up that offers a more practical algorithm. 

\section{Conclusion}

In Section~\ref{sec:threeagents} we have shown that \textsc{Proportionality by Item Deletion} is polynomial-time solvable if there are only three agents. On the other hand, if the number of agents is unbounded, then PID becomes \textsf{NP}-hard, 
and practically intractable already when we want to delete only a small number of items, as shown by the $\mathsf{W[3]}$-hardness result of Theorem~\ref{thm-w3-hardness}. 

The complexity of PID remains open for the case when the number of agents is a constant greater than 3. Is it true that for any constant $n$, there exists a polynomial-time algorithm that solves PID in polynomial time for $n$ agents? 
If the answer is yes, then can we even find an FPT-algorithm with respect to the parameter $n$? If the answer is no (that is, if PID turns out to be $\mathsf{NP}$-hard for some constant number of agents), then can we at least give an FPT-algorithm with parameter $k$ for a constant number of agents (or maybe with combined parameter $(k,n)$)?

Finally, there is ample space for future research if we consider different control actions (such as adding or replacing items), different notions of fairness, or different models for agents' preferences.

%% Either use bibtex (recommended), 
%\bibliographystyle{abbrv}
%\bibliography{prop-control}
%% .. or use the thebibliography environment explicitly

\end{document}